\documentclass[10pt,conference]{IEEEtran}

\usepackage[english]{babel}
\usepackage{amsmath,amsfonts,amssymb,bbm,amsthm,mdwlist}
\usepackage[ansinew]{inputenc}

\newtheorem{theorem}{Theorem}

\newtheorem{proposition}[theorem]{Proposition}
\newtheorem{corollary}[theorem]{Corollary}

\newtheorem{definition}[theorem]{Definition}

\renewcommand{\theequation}{\thesection.\arabic{equation}}

\def\m{{\bf m}}
\def\ds{\displaystyle}

\def\rr{\mathbb R}

\begin{document}

\title{Entropy Measures vs. Algorithmic Information }

\author{
\authorblockN{Andreia Teixeira}
\authorblockA{LIACC - U.Porto\\
andreiasofia@ncc.up.pt}
\and
\authorblockN{Armando Matos}
\authorblockA{LIACC - U.Porto\\
acm@dcc.fc.up.pt}
\and
\authorblockN{Andr\'{e} Souto}
\authorblockA{LIACC - U.Porto\\
andresouto@dcc.fc.up.pt}
\and
\authorblockN{Lu\'{i}s Antunes}
\authorblockA{LIACC - U.Porto\\
lfa@dcc.fc.up.pt}
}
\maketitle

\begin{abstract}
We further study the connection between Algorithmic Entropy and Shannon and Rényi Entropies. It is given an example for which the
difference between the expected value of algorithmic entropy and Shannon Entropy meets the known upper-bound and, for Rényi Entropy, proving that all other values of the
parameter ($\alpha$), the same difference can be big. We also prove that for a particular type of distributions Shannon Entropy is able to capture the notion
of computationally accessible information by relating it to time-bounded algorithmic entropy. In order to better
study this unexpected relation it is investigated the behavior of the
different entropies (Shannon, Rényi and Tsallis) under the
distribution based on the time-bounded algorithmic entropy.
\end{abstract}

\section{Introduction}
Algorithmic Entropy, the size of the smallest program that
generates a string, denoted by $K(x)$, is a rigorous measure of the
amount of information, or randomness, in an individual object $x$. 
Algorithmic entropy and Shannon entropy are conceptually
very different, as the former is based on the size of programs and the
later in probability distributions. Surprisingly, they are, however,
closely related. The expectation of the algorithmic entropy equals (up to a constant depending on the distribution) the Shannon entropy. 

Shannon entropy measures the amount of information in situations where
unlimited computational power is available. However this measure does
not provide a satisfactory framework for the analysis of public key
cipher systems which are based on the limited computational power of the adversary. The public key and the cipher text together
contain all the Shannon information concerning the plaintext, but the
information is computationally inaccessible. So, we face this
intriguing question: what is accessible information?

By considering the time-bounded algorithmic entropy (length of the program limited to run in time $t(|x|)$) we can take into account the
computational difficulty (time) of extracting information.
Under some computational restrictions on the distributions we
show (Theorem \ref{promise}) that Shannon entropy equals (up to a constant that depends only
on the distribution) the time-bound algorithmic information. This
result partially solves, for this type of distributions, the
problem of finding a measure that captures the notion of computationally
accessible information. This result is unexpected since it
states that for the class of probability distribution such that its cumulative
probability distribution is computable in time $t(n)$, the Shannon
entropy captures the notion of computational difficulty of extracting
information within this time bound.

With this result in mind we further study the
relation of the probability distribution based on time-bounded algorithmic entropy with several entropy measures (Shannon, Rényi and Tsallis).

\section{Preliminaries}

All strings used are elements of $\Sigma^*= \{0, 1\}^*$. $\Sigma^n$ denotes the set of strings of length $n$ and $|.|$ denotes the length of a string. It is assumed that all strings are ordered by lexicographic ordering. When $x-1$ is written, where $x$ is a string, it means the predecessor of $x$ in the lexicographic order. The function $\log$ is the function $\log_2$. The real interval between $a$ and $b$, including $a$ and excluding $b$ is represented by $[a,b)$.

\subsection{Algorithmic Information Theory}

We give essential definitions and basic results which will be need in the rest of the paper. A more detailed reference is \cite{LiVi}. The model of computation used is the prefix free Turing machine. A set of strings $A$ is prefix-free if no string in $A$ is prefix of another string of $A$. Notice that Kraft inequality guarantees that for any prefix-free set $A$, $\ds{\sum_{x\in A}2^{-|x|}\leq 1}$. 

\begin{definition}
Let $U$ be a fixed prefix free universal Turing machine. For any
string $x\in \Sigma^*$, the Kolmogorov complexity or algorithmic
entropy of $x$ is $K(x)=
\min_{p}\{|p| : U(p)= x \}$.

For any time constructible $t$, the $t$-time-bounded algorithmic entropy (or $t$-time-bounded Kolmogorov complexity) of $x\in\Sigma^*$ is, $\ds{K^{t}(x)}$ $ =\min_p\{|p| : U(p) = x \hbox{ \it in at most } t(|x|)  \hbox{  \it steps}\}$.
\end{definition}
The choice of the universal Turing
machine affects the running time of a program at most by a logarithmic
factor and the program length at most a constant number of extra
bits. 

\begin{proposition}
For all $x$ and $y$ we have:
\begin{enumerate}
	\item $K(x)\leq K^t(x)\leq |x|+O(1)$;
	\item $K(x|y)\leq K(x)+O(1)$ and $K^t(x|y)\leq K^t(x)+O(1)$;
\end{enumerate}
\end{proposition}

\begin{definition}
A string $x$ is said algorithmic-random or Kolmogorov-random if $K(x)\geq |x|$.
\end{definition}

A simple counting argument shows the existence of algorithmic-random strings of any length.

\begin{definition}
A semi-measure over a space $X$ is a function $f:X\to[0,1]$ such that $\ds{\sum_{x\in X}f(x)\leq 1}$. We say that a semi-measure is a measure if the equality holds. A semi-measure is called constructive if it is semi-computable from below.
\end{definition}

The function $\m(x)=2^{-K(x)}$ is a semi-measure which is constructible and dominates any other constructive semi-measure $\mu$ (\cite{Lev74} and \cite{Gac74}), in the sense that there is a constant $c_\mu=2^{K(\mu)}$ such that for all $x$, $\m(x)\geq c_\mu\mu(x)$. For this reason, this semi-measure is called universal. Since it is natural to consider time bounds on the Kolmogorov complexity we can define a time bounded version of $\m(x)$.

\begin{definition}
The $t$-time bounded universal distribution, denoted by $\m^t$ is $\m^t(x) = c2^{-K^t(x)}$, where $c$ is a fixed constant such that $\ds{\sum_{x\in\Sigma^*}\m^t(x)=1}$.
\end{definition}

In \cite{LiVi}, Claim 7.6.1, the authors prove that $\m^{t(\cdot)}$ dominates
every distribution $\mu$ such that $\mu^*$, the cumulative probability
distribution of $\mu$, is computable in time $t(\cdot)$.

\begin{theorem}~\label{mt}
If $\mu^*$ is computable in time $t(n)$ then there exists a constant
$c$ such that, for all $x\in\Sigma^*$, $\m^{nt(n)}(x)\geq 2^{-K^{nt(n)}(\mu)}\mu(x)$.
\end{theorem}

\subsection{Entropies}

We consider several types of entropies. Shannon information theory was introduced in 1948 by C.E. Shannon \cite{Sha48}. Information theory quantifies the uncertainty about the results of an experiment. It is based on the concept of entropy which measures the number of bits necessary to describe an outcome from an ensemble. 

\begin{definition}[Shannon Entropy \cite{Sha48}]
Let $\mathcal{X}$ be a finite or infinitely countable set and let $X$ be a random variable taking values in $\mathcal{X}$ with distribution $P$. The \emph{Shannon Entropy of random variable} $X$ is given by
$$H(X)=-\sum_{x \in \mathcal{X}}P(x) \log P(x).$$ 
\end{definition}  

The Rényi entropy is a generalization of Shannon entropy. Formally the Rényi entropy is defined as follows:

\begin{definition}[Rényi Entropy \cite{Ren61}]
Let $\mathcal{X}$ be a finite or infinitely countable set and let $X$ be a random variable taking values in $\mathcal{X}$ with distribution $P$ and let $\alpha\not=1$ be a positive real number. The Rényi Entropy of order $\alpha$ of the random variable $X$ is defined as:
\[H_\alpha(P)=\frac{1}{1-\alpha}\log\left(\sum_{x\in\mathcal{X}}P(x)^\alpha\right).\]
\end{definition}

It can be shown that $\ds{\lim_{\alpha\to1}H_\alpha(X)=H(X)}$. 

\begin{definition}[Min-Entropy]
Let $\mathcal{X}$ be a finite or infinitely countable set and let $X$ be a random variable taking values in $\mathcal{X}$ with distribution $P$. We define the Min-Entropy of $P$ by:
\[
H_\infty(P)=-\log\max_{x\in\mathcal X}P(x).
\]
\end{definition}
It is easy to see that $\ds{H_\infty(P)=\lim_{\alpha\to\infty}H_\alpha(P)}$.

\begin{definition}[Tsallis Entropy \cite{tsallis}]
Let $\mathcal{X}$ be a finite or infinitely countable set and let $X$ be a random variable taking values in $\mathcal{X}$ with distribution $P$ and let $\alpha\not=1$ be a positive real number. The Tsallis Entropy of order $\alpha$ of the random variable $X$ is defined as:
$$T_{\alpha}(P)= \frac{\ds{1-\sum_{x\in\Sigma*} P(x)^{\alpha}}}{\alpha-1}.$$
\end{definition}

\subsection{Algorithmic Information vs. Entropy Information}

Given the conceptual differences in the definition of Algorithmic
Information Theory and Information Theory, it is surprising that under some weak restrictions on the distribution of the
strings, they are closely related, in the sense that the expectation of the
algorithmic entropy equals the entropy of the distribution up to a constant that
depends only on that distribution.

\begin{theorem}\label{ke}
Let $P(x)$ be a recursive probability distribution. Then:
\[
0\leq \sum_xP(x)K(x)-H(P)\leq K(P)
\]
\end{theorem}
\proof(Sketch, see \cite{LiVi} for details) The first inequality follows directly from the well known Noiseless Coding Theorem, that, for this distributions, states
\[H(P)\leq \sum_xP(x)K(x)\]
Since $\m$ is universal, $P(x)\leq 2^{K(P)}\m(x)$, for all $x$, which is equivalent to $\log P(x) \leq K(P)-K(x)$. Thus, we have:
$$\begin{array}{c}\ds{\sum_x P(x)K(x)}-H(P) = \ds{\sum_x (P(x)(K(x)+\log P(x)))}\\
\leq \ds{\sum_x (P(x)(K(x)+K(P)-K(x)))}= K(P)\;\;\;\;\Box
\end{array}$$

\section{Algorithmic Entropy vs. Entropy: How Close?}

Given the surprising relationship between algorithmic entropy and entropy, in this section we investigate how close they are. We study
also the relation between algorithmic entropy and Rényi
entropy. In particular, we will find the values of $\alpha$ for which the  same relation as in Theorem
\ref{ke} holds for the Rényi entropy. We also prove that for a
particular type of distributions, entropy is able to capture the notion
of computationally accessible information. 

First we show that the interval $[0,K(P)]$ of the inequalities of Theorem \ref{ke} is tight:

\begin{proposition}
There exist distributions $P$, with $K(P)$ large such that:

\begin{enumerate}
	\item $\ds{\sum_x P(x)K(x) - H(P) = K(P)-O(1)}$.
	\item $\ds{\sum_x P(x)K(x) - H(P)=O(1)}$.
\end{enumerate}
\end{proposition}

\begin{proof}
\begin{enumerate}
	\item Fix $x_0\in\Sigma^n$. Consider the following probability distribution:
$$P_n(x)= \left\{
\begin{array}{lll}
1 & \textnormal{if } x= x_0\\
0 & \textnormal{otherwise} 
\end{array}
\right.$$

Notice that describing the distribution is equivalent to describe $x_0$. So, $K(P_n)= K(x_0)+O(1)$. On the other hand, $\sum_x P_n(x)K(x) - H(P_n) = K(x_0)$. So, if $x_0$ is Kolmogorov-random then $K(P_n)\approx n$.

\item Let $y$ be a string of length $n$ such that $K(y)=n-O(1)$ and consider the following probability distribution over $\Sigma^*$:
\[
P_n(x)=\left\{ \begin{array}{ll}

							0.y & \textnormal{if $x=x_0$}\\
							
							1-0.y & \textnormal{if $x=x_1$}\\
							
							0 & \textnormal{otherwise}

\end{array}\right.\]
where $0.y$ represent the real number between 0 and 1 which binary representation is $y$. Notice that we can choose $x_0$ and $x_1$ such that $K(x_0)=K(x_1)\leq c$ where $c$ is a constant that does not depend on $n$.

Thus we have:

\begin{enumerate}
	\item $K(P_n)\approx n$, since describing $P_n$ is equivalent to describe $x_0$, $x_1$ and $y$;
	
	\item $\ds{\sum_x P_n(x)K(x)=(0.y)K(x_0)+(1-0.y) K(x_1)}$ ${\leq 0.y\times c+(1-0.y)\times c=c}$;
	
	\item $H(P_n)=- 0.y\log 0.y -(1-0.y)\log(1-0.y)\leq 1$
\end{enumerate}

Thus $\sum P_n(x)K(x)-H(P_n)\leq c<\!<K(P_n)\approx n$.
\end{enumerate}
\end{proof}

Now we address the question if the same relations as in Theorem \ref{ke} holds for the Rényi entropy. We show, in fact, the Shannon entropy is the ``smallest'' entropy that verify these properties.

Since, for every $0<\varepsilon<1$, 
$$H_\infty \leq H_{1-\varepsilon} (X) \leq H(X) \leq H_{1+\varepsilon}(X) \leq H_0(X)$$
it follows that
$$\ds{0 \underbrace{\leq}_{\alpha \geq 1} \sum_x P(x)K(x) - H_\alpha(P) \underbrace{\leq}_{\alpha \leq 1} K(P)}$$ 

In the next result we show that the inequalities above are, in general, false for different values of $\alpha$.

\begin{theorem}
For every $\Delta>0$ and $\alpha>1$ there exists a recursive distribution $P$ such,
\begin{enumerate}
	\item $\ds{\sum_x P(x)K(x) - H_\alpha(P) \geq (K(P))^\alpha}$
	\item $\ds{\sum_x P(x)K(x) - H_\alpha(P) \geq K(P)+\Delta}$
\end{enumerate}
\end{theorem}

The proof of this Theorem is similar to the proof of the following Corollary:

\begin{corollary}
There exists a recursive probability distribution $P$ such that:
\begin{enumerate}
	\item $\ds{\sum_x P(x)K(x) - H_{\alpha}(P) > K(P)}$, where $\alpha>1$;
	\item $\ds{\sum_x P(x)K(x) - H_{\alpha}(P) < 0}$, where $\alpha<1$.
\end{enumerate}
\end{corollary}

\proof
For $x \in \{0,1\}^n$, consider the following probability distribution:
$$P_n(x)= \left\{
\begin{array}{lll}
1/2 & \textnormal{if } x=0^n\\
2^{-n} & \textnormal{if } x=1x', x' \in \{0,1\}^{n-1}\\
0 & \textnormal{otherwise} 
\end{array}
\right.$$
	
It is clear that this distribution is recursive.

\begin{enumerate}
	\item First observe that
$$\begin{array}{lllll}
H(P_n) &= -\ds{\sum_x P_n(x)\log P_n(x)}\\
&=\ds{-\left(\frac{1}{2}\log \frac{1}{2}+\frac{1}{2^n}\log \frac{1}{2^n}2^{n-1}\right)}\\
&\ds{=-\left(-\frac{1}{2}-n\frac{1}{2^n}2^{n-1}\right)}\\
&=\ds{\frac{n+1}{2}}
\end{array}$$

Notice also that $K(P_n) = O(\log n)$.
	
We want to prove that, for every $\alpha>1$,
$$( \exists n_0) (\forall n \geq n_0) \sum_x P_n(x)K(x) - H_\alpha(P_n) > K(P_n)$$

Fix $\alpha$ such that $\alpha - 1 = \frac{1}{(n_0-1)^{1.8}}$.  
\begin{eqnarray*}
H_\alpha(P_n) & = & \frac{1}{1-\alpha}\log \ds{\sum_x P_n(x)^\alpha}\\         
            & = & \frac{1}{1-\alpha}\log\left(\frac{1}{2^\alpha}+2^{n-1}\times \frac{1}{2^{n\alpha}}\right)\\
            & = & \frac{1}{1-\alpha}\left(\log (2^{(n-1)\alpha}+2^{n-1})-n\alpha\right)
\end{eqnarray*}           
            
Now we calculate $\log \left(2^{(n-1)\alpha}+2^{n-1}\right)$. To simplify notation consider:
$$\left\{
\begin{array}{lll}
x & = & n-1\\
\alpha & = & 1+\varepsilon, \textnormal{with } \varepsilon > 0
\end{array}
\right.$$

Thus, 
$$\begin{array}{ll}\log \left(2^{(n-1)\alpha}+2^{n-1}\right) &= \log \left(2^{x(1+\varepsilon)}+2^x\right)=\\ 
&= \log \left(2^x\left(2^{x\varepsilon}+1\right)\right)\\
&= x + \log \left(2^{x\varepsilon}+1\right)
\end{array}$$

Consider $ \delta =x\varepsilon$. It is clear that
$$2^\delta = e^{\ln 2\cdot \delta} = 1+\ln 2 \cdot \delta + \frac{(\ln 2)^2\cdot \delta^2}{2}+\cdots$$
then,
$$2^\delta + 1 = 2 + \ln 2 \cdot \delta + \frac{(\ln 2)^2\cdot \delta^2}{2}+\cdots$$
and hence,
$$\log (2^\delta + 1) = \log \left(2 + \underbrace{\ln 2 \cdot \delta + \frac{(\ln 2)^2\cdot \delta^2}{2}+\cdots}_{\beta}\right)$$

Notice that $\ds{\lim_{\alpha\to 1}\beta=0}$.
\[\begin{array}{ll}\log (2 + \beta) &= \ds{\frac{1}{\ln 2}\ln(2+\beta)}\\
&\ds{= \frac{1}{\ln 2}\ln(2(1+\frac{\beta}{2}))=}\\
&\ds{= \frac{1}{\ln 2}(\ln 2 + \frac{\beta}{2}-\frac{\beta^2}{8}+\cdots)}\\
&\ds{= 1+ \frac{\beta}{2\ln 2}-\frac{\beta^2}{8\ln 2}+\cdots}\end{array}\]
Then,
$$\begin{array}{l}
\ds{\log (2 + \ln 2 \cdot \delta + \frac{(\ln 2)^2\cdot \delta^2}{2}+\cdots)} = \\= \ds{ 1 + \frac{\delta}{2}+\frac{\ln 2}{8}\delta^2 + \cdots - \frac{(\ln 2)^2}{8}\delta^3 - \frac{(\ln 2)^3}{32}\delta^4-\cdots}
\end{array}$$

So we have:
\[\log (2^{x\varepsilon}+1) = 1 + \frac{x\varepsilon}{2}+\frac{\ln 2}{8}(x\varepsilon)^2 + \cdots\]
which means,
$$x+\log (2^{x\varepsilon}+1) = x + 1 + \frac{x\varepsilon}{2}+\frac{\ln 2}{8}(x\varepsilon)^2 + \cdots$$

Thus
$$\begin{array}{ll}
H_\alpha(P_n) & = \ds{\frac{-1}{\alpha-1}(\log (2^{(n-1)\alpha}+2^{n-1})-n\alpha)}\\
& =\ds{n - \frac{n-1}{2}-\frac{\ln 2}{8}(n-1)^2(\alpha-1)-\cdots}
\end{array}$$

Notice that the rest of elements in the series expansion
$$c_1(n-1)^3(\alpha-1)^2+c_2(n-1)^4(\alpha-1)^3+\cdots, c_1,c_2 \in \mathbb{R}$$
can be ignored in the limit since $\alpha-1 = \frac{1}{(n_0-1)^{1.8}}$.

So, for all $n\geq n_0$:
$$\begin{array}{lll}
H_\alpha(P_n) & = & \ds{\frac{n+1}{2}-\frac{\ln 2}{8}(n-1)^{0.2}}
\end{array}$$

It is known that $\ds{\lim_{\alpha \to 1}H_\alpha(P_n)=H(P_n)}$. In fact, we have $H_\alpha(P_n) = H(P_n) - \frac{\ln 2}{8}(n_0-1)^{0.2}$.

Now, the first item of the Theorem is proved by contradiction. Assume by contradiction that
$$\ds{\sum_x P_n(x)K(x)}-H_\alpha(P_n) \leq c\log n , \textnormal{with $c \in \mathbb{R}$}$$
i.e., for all $n\geq n_0$
$$\ds{\sum_x P_n(x)K(x)}- H(P_n)+\frac{\ln 2}{8}(n-1)^{0.2} \leq c\log n$$

Since, $\ds{\sum_x P_n(x)K(x)}- H(P_n) \geq 0$, we would have $\frac{\ln 2}{8}(n-1)^{0.2} \leq c\log n$, which is impossible for all $n\geq n_0$. So, we conclude that

$$\ds{\sum_x P_n(x)K(x)}-H_\alpha(P_n) > c\log n.$$
	
	\item Analogous to the proof of the previous item, but now fixing $\ds{\alpha - 1 = \frac{-1}{(n-1)^{1.8}}}$.\hfill$\Box$
%
%
\end{enumerate}
\ \\

If instead of considering $K(P)$ and $K(x)$ in the inequalities of
Theorem \ref{ke} we use the time bounded version and imposing some
computational restrictions on the distributions we obtain a similar
result. Notice that for the class of distributions on the following
Theorem the entropy equals (up to a constant) the time-bounded
algorithmic entropy. 
\begin{theorem}\label{promise}
Let $P$ be a probability distribution such that $P^*$, the cumulative
probability distribution of $P$, is computable in time $t(n)$. Then:
$$\ds{0\leq\sum_x P(x)K^{nt(n)}-H(P)\leq K^{nt(n)}(P)}$$
\end{theorem}
\proof
The first inequality follows directly from Theorem \ref{ke} and from the fact that $K^t(x)\geq K(x)$.

By Theorem \ref{mt}, if $P$ is a probability distribution such that $P^*$ is computable in time $t(n)$, then for all $x\in\Sigma^n$ 
$$K^{nt(n)}(x) + \log P(x) \leq K^{nt(n)}(P)$$
Then, summing over all $x$ we get 
$$\ds{\sum_x P(x)(K^{nt(n)}(x) + \log P(x)) \leq \sum_x P(x)K^{nt(n)}(P)}$$
which is equivalent to
$$\ds{\sum_x P(x)K^{nt(n)}(x) - H(P) \leq  K^{nt(n)}(P)}\;\;\;\Box$$

This result partially solves, for this type of distributions, the
problem of finding a measure that captures the notion of computationally
accessible information. This is an important open problem with
applications and consequences in cryptography.

\section{On the entropy of the time-bounded algorithmic universal distribution}

We now focus our attention on the universal distribution. Its main
drawback is the fact that it is not computable. In order to make it computable, one can impose restrictions on the time that a program can use to produce a string obtaining the time-bounded universal distribution ($\m^t(x)=c2^{-K^t(x)}$). We
investigate the behavior of the different entropies under this
distribution. The proof of the following Theorem uses some ideas from
\cite{KT}.

\begin{theorem}\label{hmt}
The Shannon entropy of the distribution $\m^t$ diverges.
\end{theorem}

\begin{proof}
If $x\geq 2$ then $f(x)=x2^{-x}$ is a decreasing function. Let $A$ be the set of strings such that $-\log \m^t(x)\geq 2$. Since $\m^t$ is computable, $A$ is recursively enumerable. Notice also that $A$ is infinite and contains arbitrarily large Kolmogorov-random strings.
\[\begin{array}{c}
\ds{\sum_{x\in\Sigma^*} - \m^t(x)\log \m^t(x)\geq \sum_{x\in A} - \m^t(x)\log \m^t(x)}\\
\ds{=\sum_{x\in A}c2^{-K^t(x)}(K^t(x) -\log c)}\\
=\ds{-c\log c\sum_{x\in A}2^{-K^t(x)}+c\sum_{x\in A}K^t(x)2^{-K^t(x)}
}\end{array}\]
So if we prove that $\ds{\sum_{x\in A}K^t(x)2^{-K^t(x)}}$ diverges the result follows.

Assume, by contradiction, that $\ds{\sum_{x\in A}K^t(x)2^{-K^t(x)}<d}$ for some $d\in\rr$. Then, considering $r(x)=\ds{\frac{1}{d}K^t(x)2^{-K^t(x)}}$ if $s\in A$ and $r(x)=0$ otherwise, we conclude that $r$ is a semi-measure. Thus, there exists a constant $c'$ such that, for all $x$, $r(x)\leq c'\m(x)$. Hence, for $x\in A$, we have
\[
\frac{1}{d}K^t(x)2^{-K^t(x)}\leq c'2^{-K(x)}
\]

So, $K^t(x)\leq c'd2^{K^t(x)-K(x)}$. This is a contradiction since $A$ contains Kolmogorov - random strings of arbitrarily large size. The contradiction results from assuming that $\ds{\sum_{x\in A}K^t(x)2^{-K^t(x)}}$ converges. So, $H(\m^t)$ diverges.
\end{proof}
Now we show that, similarly to the behavior of entropy of universal distribution, $T_\alpha(\m^t)<\infty$ iff $\alpha>1$ and  $H_\alpha(\m^t)<\infty$ iff $\alpha<1$. First obverse that we have the following ordering relationship between these two entropies for all probability distribution $P$:
\begin{enumerate}
	\item If $\alpha>1$, $\ds{T_\alpha(P)\leq\frac{1}{\alpha-1}+H_\alpha(P)}$;
	\item If $\alpha<1$, $\ds{T_\alpha(P)\geq\frac{1}{\alpha-1}+H_\alpha(P)}$;
	
\end{enumerate}

\begin{theorem}
Let $\alpha\not=1$ be a real computable number. Then we have, $T_\alpha(\m^t)<\infty$ iff $\alpha>1$.
\end{theorem}

\begin{proof}
From Theorem 8 of \cite{KT}, it is known that $\ds{\sum_{x\in\Sigma^*}(\m(x))^\alpha}$ converges iff $\alpha>1$. Since $\m^t$ is a probability measure there exists a constant $\lambda$ such that, for all $x$, $\m^t(x)\leq \lambda\m(x)$. So, $(\m^t(x))^\alpha\leq (\lambda\m(x))^\alpha$, which implies that $\ds{\sum_{x\in\Sigma^*}(\m^t(x))^\alpha\leq\lambda^\alpha\sum_{x\in\Sigma^*}(\m(x))^\alpha}$, from where we conclude that, for $\alpha>1$,  $T_\alpha(\m^t)$ converges.

For $\alpha<1$, the proof is analogous to the proof of Theorem \ref{hmt}.
Suppose that $\ds{\sum_{x\in\Sigma^*}(\m^t(x))^\alpha<d}$ for some $d\in\rr$. Hence, $r(x)=\ds{\frac{1}{d}(\m^t(x))^\alpha}$ is a computable semi-measure. Then, there exists a constant $\tau$ such that for all $x\in\Sigma^*$, $r(x)=\ds{\frac{1}{d}(c2^{-K^t(x)})^\alpha\leq \tau 2^{-K(x)}}$ which is equivalent to $\ds{\frac{c^\alpha}{d\tau}\leq 2^{\alpha K^t(x)-K(x)}}$. For example, if $x$ is random it follows that $\ds{\frac{c^\alpha}{d\tau}\leq 2^{(\alpha-1)|x|}}$, which is false.
\end{proof}

\begin{theorem}
The Rényi entropy of order $\alpha$ of time bounded universal distribution converges for $\alpha<1$ and diverges if $\alpha>1$.
\end{theorem}

\begin{proof}
Consider $\alpha=1+\varepsilon$, where $\varepsilon>0$. Since for all $x\in\Sigma^*$, $\ds{K^t(x)}  \leq  |x| + c'$ then 
$2^{-|x|+c} \leq 2^{-K^t(x)}$. Since $f(y)=y^{1+\varepsilon}$ increases in $[0,1]$, it is also true that for all $x\in\Sigma^*$, $(2^{-|x|+c})^{1+\varepsilon} \leq(2^{-K^t(x)})^{1+\varepsilon}$. So, summing up over all $x\in\Sigma^*$ and applying $-\log$ we conclude that 
$$
 - \log \sum_x(2^{-K^t(x)})^{1+\varepsilon}  \leq  - \log \sum_x(2^{-|x|+c})^{1+\varepsilon}
$$

If we prove that the series $\ds{\sum_{x\in\Sigma^*}(2^{-|x|+c})^{1+\varepsilon}}$ converges, then the Rényi entropy of order $1+\varepsilon$ of $\m^t$ also converges.
$$\begin{array}{ll}
\ds{\sum_{x\in\Sigma^*}(2^{-|x|+c})^{1+\varepsilon}}&=\ds{\sum_{n=1}^{\infty}\sum_{x \in \Sigma^n}(2^{-n+c})^{1+\varepsilon}}\\
  &=  \ds{\sum_{n=1}^{\infty}\sum_{x \in \Sigma^n}2^{-n-n\varepsilon+c+c\varepsilon}}\\
&=\ds{\sum_{n=1}^{\infty} 2^n \times 2^{-n-n\varepsilon} \times 2^{c+c\varepsilon}}\\  &=  \ds{2^{c+c\varepsilon}\sum_{n=1}^{\infty} 2^{-n\varepsilon}}\\
 &= \ds{ 2^{c+c\varepsilon}\times \frac{2^\varepsilon}{2^\varepsilon-1}<\infty}
\end{array}$$  

Now, assume that $\alpha<1$. Since the  Rényi entropy is non increasing with $\alpha$, for any distribution $P$ we have $H(P) \leq H_\alpha(P)$. So, in particular, $H(\m^t)\leq H_\alpha(\m^t)$. As $H(\m^t)$ diverges we conclude that the Rényi entropy of order $\alpha < 1$ for the time bounded universal distribution diverges.
\end{proof}
%
%
%
%
%
%
%

\end{document}